\documentclass[11pt]{article}
\usepackage[utf8]{inputenc}
\usepackage{hyperref}
\usepackage{amsmath}
\usepackage{amsthm}
\usepackage{amssymb}
\usepackage{authblk}
\usepackage{cite}
\usepackage{lipsum}
\usepackage{geometry}
\usepackage{enumerate}
\usepackage{array}
\usepackage{multicol}
\usepackage{multirow}
\usepackage{xcolor}
\newtheorem{theorem}{Theorem}[section]
\numberwithin{theorem}{section}
\newtheorem{lemma}[theorem]{Lemma}
\newtheorem{defi}[theorem]{Definition}
\newtheorem{corollary}[theorem]{Corollary}
\newtheorem{prop}[theorem]{Proposition}
\newtheorem{remark}[theorem]{Remark}

\DeclareMathOperator{\tr}{Tr}

\newcommand{\Max}{\displaystyle \max}

\newcommand{\Frac}{\displaystyle \frac}
\newcommand{\F}{\mathbb F}
\newcommand{\Fbn}{\mathbb{F}_{2^n}}

\newcommand{\Fpn}{\mathbb{F}_{p^n}}

\newcommand{\Fbnmul}{\mathbb{F}_{2^n}^*}
\newcommand{\Fpnmul}{\mathbb{F}_{p^n}^*}

\begin{document}

    \title{The second-order zero differential uniformity of the swapped inverse functions over finite fields}
    \author{Jaeseong Jeong$^1$, Namhun Koo$^2$, Soonhak Kwon$^3$\\
        \small{\texttt{ Email: wotjd012321@naver.com, nhkoo@ewha.ac.kr, shkwon@skku.edu}}\\
\small{$^1$Department of Innovation Center for Industrial Mathematics, National Institute for Mathematical Sciences, Seongnam, Republic of Korea}\\
\small{$^2$Institute of Mathematical Sciences, Ewha Womans University, Seoul, Republic of Korea}\\
\small{$^3$Department of Mathematics, Sungkyunkwan University, Suwon, Republic of Korea}\\
    }

\maketitle
\begin{abstract}
The Feistel Boomerang Connectivity Table (FBCT) was proposed as the feistel counterpart of the Boomerang Connectivity Table. The entries of the FBCT are actually related to the second-order zero differential spectrum. Recently, several results on the second-order zero differential uniformity of some functions were introduced. However, almost all of them were focused on power functions, and there are only few results on non-power functions. In this paper, we investigate the second-order zero differential uniformity of the swapped inverse functions, which are functions obtained from swapping two points in the inverse function. We also present the second-order zero differential spectrum of the swapped inverse functions for certain cases. In particular, this paper is the first result to characterize classes of non-power functions with the second-order zero differential uniformity equal to $4$, in even characteristic.


\end{abstract}

\section{Introduction}

Throughout this paper, let 
\begin{itemize} 
\item $\Fpn$ be the finite field of $p^n$ elements and $\Fpnmul=\Fpn \setminus \{0\}$ be the multiplicative group,
where $p$ is a prime.
\item $Inv(x) = x^{p^n-2}$ be the multiplicative inverse function on $\Fpn$.
\end{itemize}

The substitution box(S-box) plays an important role in the security of block ciphers, since it is the only nonlinear part in block ciphers. Hence, vectorial Boolean functions to be used as S-boxes in block ciphers should have good cryptographic properties. For example, S-boxes need to have low differential uniformity, where the differential uniformity\cite{Nyb94} is a cryptographic parameter to measure the resistance against differential attack, which is one of the most well-known attacks on block ciphers. The definition of the differential uniformity is as follows.
\begin{defi}Let $f$ be a function over $\Fpn$ and $a,b\in\Fpn$. Then we define $\Delta_f(a,b)$ by the number of solutions of 
$f(x+a)-f(x)=b$. The differential uniformity of $f$ is defined by
\begin{equation*}
\Delta_f=\Max_{a\in\Fpnmul, b\in \Fpn} \Delta_f(a,b).
\end{equation*}
Furthermore, we say that $f$ is second-order zero differentially $\Delta_f$-uniform. We say $f$ is Perfect Nonlinear (PN) if $\Delta_f=1$, and $f$ is Almost Perfect Nonlinear (APN) if $\Delta_f=2$. 
\end{defi}

The boomerang attack, an extension of the differential attack, was proposed by Wagner\cite{Wag99}. At Eurocrypt 2018, Cid et al. \cite{CHP+18} introduced the Boomerang Connectivity Table (BCT), as a tool to analyze the boomerang attack. The BCT only addressed the case of Substitution Permutation Networks, and does not cover the case of block ciphers with feistel structures. Boukerrou et al.\cite{BHL+20} proposed the Feistel Boomerang Connectivity Table (FBCT), the feistel counterpart of the BCT. They also introduced the Feistel Boomerang Uniformity, which is defined by the maximum value in the FBCT except the first row, the first column, and diagonal part. In \cite{LYT22}, Li et al. observed that the entries of the FBCT are actually related to the second-order zero differential spectra of functions. The definition of the second-order zero differential uniformity is given as follows.

\begin{defi} Let $f$ be a function over $\Fpn$ and $a,b\in\Fpn$. Then we define $\nabla_f(a,b)$ by the number of solutions of 
\begin{equation}\label{FBU_eqn}
 f(x+a+b)-f(x+a)-f(x+b)+f(x)=0.
\end{equation}
We also define the second-order zero differential uniformity by
\begin{equation*}
\nabla_f=
\begin{cases}
\Max \{\nabla_f(a,b) : a,b\in\Fbnmul, a\ne b \} &\text{ if }p=2,\\
\Max \{\nabla_f(a,b) : a,b\in\Fpnmul \} &\text{ if }p>2.
\end{cases}
\end{equation*}
Furthermore, we say that $f$ is second-order zero differentially $\nabla_f$-uniform. We denote 
$$\omega_i=\#\{(a,b)\in\Fpn\times\Fpn : \nabla_f(a,b)=i\},$$
for $i\le \nabla_f$. Then, the second-order zero differential spectrum of $f$ by
$$DS_f = \{\omega_i\ :\ i\le \nabla_f \}.$$
\end{defi}

Recently, several researchers studied functions with low second-order zero differential uniformity\cite{EM22,GHRS23a,GHRS23b,LYT22,MML+23,MLXZ23}. The majority of those functions are power functions, and there are only few results on the second-order zero differential uniformity of non-power functions.
Table \ref{table_nonpower} describes all known non-power functions for which the second-order zero differential uniformity is known.

\begin{table}
\begin{center}
\begin{tabular}{|c|c|c|c|c|}
\hline $F(x)$ & $p$ & Conditions & $\nabla_f$ & Ref. \\
\hline $x^{2^n-2} + \tr \left(\frac{x^2}{x+1}\right)$ & $2$ & $n$ : even & $8$ & \multirow{2}{*}{\cite{GHRS23a}}\\
\cline{1-4} $1/\left(x+\gamma \tr(x^{2^k+1})\right)$ & $2$ & $k<n$, $\gamma \in \Fbn \cap \F_{2^{2k}}^*$, $\tr(\gamma^{2k+1})=0$ & $\le 8$ & \\
\hline \multirow{2}{*}{$x^{p^n-1}+ux^2$} & $p=3$ & \multirow{2}{*}{$u\ne 0$} & $2$ & \multirow{2}{*}{\cite{GHRS23b}}\\
\cline{2-2} \cline{4-4} & $p>3$ &  & $4$ & \\
\hline
\end{tabular}
\caption{Known non-power functions studied for their second-order zero differential uniformity}\label{table_nonpower}
\end{center}
\end{table}

For a transposition $(\alpha,\beta)$ on $\Fpn$ where $\alpha, \beta \in \Fpn$ with $\alpha \ne \beta$, we call functions of the form $Inv \circ (\alpha, \beta)$ by the \emph{swapped inverse functions}. Many researchers were interested in cryptographic properties of the swapped inverse functions, and it is known that the swapped inverse functions have good cryptographic properties, for example, low differential uniformity\cite{LWY13e,JKK23}, low $c$-differential uniformity\cite{Sta21,JKK23}, low boomerang uniformity\cite{LQSL19,CV20}, low $c$-boomerang uniformity\cite{Sta21}, low differential-linear uniformity\cite{JKK22} and high nonlinearity\cite{LWY13e}. In this paper, we study the second-order zero differential uniformity of the swapped inverse functions.
As observed in \cite{JKK23}, if $\alpha \ne 0$, then the swapped inverse function $F=Inv \circ (\alpha, \beta)$ is linear equivalent to  $$Inv \circ (1,\alpha^{-1}\beta) = L_\alpha \circ F \circ L_\alpha,$$ 
where $L_{\alpha}(x) =\alpha x$. Since the second-order zero differential uniformity is invariant under extended-affine equivalence\cite{BHL+20}, $F$ and $Inv\circ (1,\alpha^{-1}\beta)$ have the same second-order zero differential uniformity. Hence, in this paper, we investigate the second order zero differential uniformity of $Inv \circ (1,\gamma)$, where $\gamma \in \Fpn \setminus\{1\}$. 

We divided our investigation into two sections: Section \ref{sec_gam0} for the case $\gamma=0$, and Section \ref{sec_gamnot0} for the case $\gamma\ne0$.  Especially when $p=2$, we study the second-order zero differential spectrum of $Inv \circ (1,\gamma)$. For $p$ odd, we find an upper bound of the second-order zero differential uniformity of $Inv\circ (0,1)$. We could not completely determine the second-order zero differential uniformity of $Inv\circ (1,\gamma)$ for $\gamma \ne 0$, but we give several experimental results via SageMath and state some conjectures on second-order zero differential uniformity of $Inv \circ (1,\gamma)$ when $p=3$. 
This case ($\gamma\ne 0$) is notably extensive, posing many challenges for a complete analysis. Similar difficulties have already been observed in the evaluation of other uniformity invariants for this general case $\gamma \ne 0$. For example, a research\cite{CV20} for the boomerang uniformity of $Inv \circ (1,\gamma)$ with $\gamma\ne 0$ required 8 pages to be addressed in its preprint version; however, this was ultimately omitted from the journal publication version. Furthermore, to the best of our knowledge, there are no known results concerning the $c$-differential uniformity or $c$-boomerang uniformity for this general case.

The rest of this paper are organized as follows. In Section \ref{sec_pre}, we give some preliminaries.  In Section \ref{sec_gam0}, we investigate the second-order zero differential uniformity of $Inv \circ (0,1)$. In Section \ref{sec_gamnot0}, we study the second-order zero differential uniformity of $Inv \circ (1,\gamma)$ when $\gamma \ne 0$. Finally, we give the concluding remark in Section \ref{sec_con}.

\section{Preliminaries}\label{sec_pre}

In the following, we give some properties of the FBCT when $p=2$. 

\begin{prop}(\cite{BHL+20})\label{FBU_prop_even}
Let $f$ be a function on $\Fbn$ and $a,b\in\Fbn$. Then,
\begin{enumerate}
\item $\nabla_f(a,b)=\nabla_f(b,a)$.
\item $\nabla_f(a,b)=2^n$ when $ab(a+b)=0$.
\item $\nabla_f(a,b)\equiv 0\pmod{4}$.
\item $\nabla_f(a,b)=\nabla_f(a,a+b)$.
\item $f$ is APN if and only if $\nabla_f(a,b)=0$ for all $a,b\in \Fbn$ with $ab(a+b)\ne 0$.
\end{enumerate}
\end{prop}
By the third term of the above proposition, we see that $\nabla_f(a,b)\equiv 0 \pmod{4}$ and hence $\omega_i=0$ for $i\not\equiv 0\pmod{4}$. The last term of the above proposition means that $f$ is APN if and only if $f$ is second-order zero differentially $0$-uniform. Hence, in studies of the second-order zero differential uniformity, only functions that are not APN have been considered, and the minimal possible second-order zero differential uniformity is $4$, in this case. In this paper, we completely characterize second-order zero differentially $4$-uniform swapped inverse functions.

Next, we give some properties of the FBCT when $p$ is an odd prime. Note that some of the above proposition does not hold in this case.

\begin{prop}(\cite{LYT22})
Let $p$ be an odd prime and $f$ be a function on $\Fpn$ and $a,b\in\Fpn$. Then,
\begin{enumerate}
\item $\nabla_f(a,b)=\nabla_f(b,a)$.
\item $\nabla_f(a,b)=p^n$ when $ab=0$.
\item $f$ is PN if and only if $\nabla_f(a,b)=0$ for all $a,b\in \Fpn$ with $ab\ne 0$.
\end{enumerate}
\end{prop}
Similarly with the case of characteristic $2$, only functions that are not PN have been considered in this research. Unlike the binary case, there are second-order zero differentially $1$-uniform functions\cite{GHRS23a,LYT22}. In this paper we show that the second-zero differential uniformity of $Inv \circ (0,1)$ is at most $4$ for odd characteristic case.

We can see that if $x$ is a solution of \eqref{FBU_eqn}, then $x+a$ is a solution of $f(x-a+b)-f(x-a)-f(x+b)+f(x)=0$, and $x+b$ is a solution of $f(x+a-b)-f(x+a)-f(x-b)+f(x)=0$. So, (especially in odd characteristic) we have the following.
\begin{equation}\label{-a-b_eqn}
\nabla_f(a,b)=\nabla_f(-a,b)=\nabla_f(a,-b)=\nabla_f(-a,-b).
\end{equation}

Next we recall that the FBCT of $Inv$.

\begin{prop}(\cite{EM22}) Let $p=2$. Then,
\begin{equation*}
\nabla_{Inv}(a,b)=
\begin{cases}
0, &\text{ if }ab(a+b)\ne 0, \frac a b\not\in \F_4 \setminus \F_2,\\
4, &\text{ if }ab(a+b)\ne 0, \frac a b \in \F_4 \setminus \F_2,\\
2^n, &\text{ if }ab(a+b)= 0.
\end{cases}
\end{equation*}

\end{prop}

Table \ref{Inv_char_odd_table} describes $\nabla_{Inv}(a,b)$ when $p$ is odd, which was studied in \cite{LYT22}. 
In \cite{LYT22}, the authors state that 
\begin{equation}\label{Inv_char_odd_eq}
\nabla_{Inv}(a,b)=
\begin{cases}
1, &\text{ if }ab\ne 0, a\ne b,\\
3, &\text{ if }ab\ne 0, (a^2+ab+b^2)(a^2-ab+b^2)=a^4+a^2b^2+b^4=0,\\
p^n, &\text{ otherwise,}
\end{cases}
\end{equation}
for the case $3\mid (p^n-1)$. But, the cases $\nabla_{Inv}(a,b)=1$ and $\nabla_{Inv}(a,b)=3$ in \eqref{Inv_char_odd_eq} do not totally cover all nontrivial cases that $ab\ne 0$, leading to a potential misunderstanding that $\nabla_{Inv}=p^n$. Hence, in our opinion, this case should be clarified as shown in Table \ref{Inv_char_odd_table}.

\begin{table}
\begin{center}
\begin{tabular}{|c|c|c|c|}
\hline $\nabla_{Inv}(a,b)$ & 1 & 3 & $p^n$\\
\hline\hline $p=3$ & $ab\ne 0$, $a^2 \ne b^2$ & $ab\ne 0$, $a^2 = b^2$ & $ab=0$\\
\hline $3\mid (p^n-1)$ & $ab\ne 0$, $a^4+a^2b^2+b^4\ne 0$ & $ab\ne 0$,  $a^4+a^2b^2+b^4=0$ & $ab=0$\\
\hline $p>3$ and $3\nmid(p^n-1)$ & $ab\ne0$ & - & $ab=0$\\
\hline
\end{tabular}
    \caption{The value $\nabla_{Inv}(a,b)$ when $p$ is odd \cite{LYT22}}\label{Inv_char_odd_table}
\end{center}
\end{table}

The following is the well-known lemma for the number of solutions of quadratic equations on $\Fbn$, which is useful for our results.
\begin{lemma}(\cite{LN97})\label{quad lemma} Let $a\in \Fbnmul$ and $b,c\in \Fbn$.
\begin{enumerate}
\item $ax^2+bx+c=0$ has one solution in $\Fbn$ if and only if $b=0$.
\item $ax^2+bx+c=0$ has two solutions in $\Fbn$ if and only if $\tr\left(\frac{ac}{b^2}\right)=0$.
\item $ax^2+bx+c=0$ has no solution in $\Fbn$ if and only if $b\ne0$ and $\tr\left(\frac{ac}{b^2}\right)=1$.
\end{enumerate}
\end{lemma}

\section{The second-order zero differential uniformity of $Inv \circ (0,1)$}\label{sec_gam0}

\subsection{The second-order zero differential spectrum in even characteristic}

In this subsection, we study the second-order zero differential spectrum of $Inv \circ (0,1)$. In the next theorem, we investigate the FBCT of $f(x) = Inv\circ (0,1)$.

\begin{theorem}\label{Inv01_char_even_thm} Let $p=2$ and $f(x) = Inv\circ(0,1)$. Then,
\begin{equation*}
\nabla_f(a,b)=
\begin{cases}
2^n, &\text{ if }ab(a+b)= 0,\\
8, &\text{ if }a\ne b, a^3+a+1=b^3+b+1=0,\\
4, &\text{ if }ab(a+b)\ne 0\text{ and }[a,b\in\F_4^*\text{ or }(a,b)\in S],\\
0, &\text{ otherwise.}
\end{cases}
\end{equation*}
where $S=\{(a,b)\in\Fbnmul\times\Fbnmul : a^2+b^2+ab\in \{1, ab(a+b)\}\}\setminus\{(a,b)\in\F_{2^3}^*\times\F_{2^3}^* : a^3+a+1=b^3+b+1=0\}$.
\end{theorem}
\begin{proof}
If $a=1$, then $b\ne 0,1$ and \eqref{FBU_eqn} is equivalent to $f(x)+f(x+1)+f(x+b)+f(x+b+1)=0$. If $x\in \{0,1,b,b+1\}$, then $f(x)+f(x+1)+f(x+b)+f(x+b+1)=0$ is equivalent to $b^2+b+1=0$, that is, $b\in \F_4 \setminus \F_2$. If $x\not\in \{0,1,b,b+1\}$, then $f(x)+f(x+1)+f(x+b)+f(x+b+1)=0$ is equivalent to $b^2+b=0$ , a contradiction. So, $\nabla_f(1,b)=4$ if and only if $b\in \F_4 \setminus \F_2$, and $\nabla_f(1,b)=0$ otherwise. By Proposition \ref{FBU_prop_even}, we can see that $\nabla_f(1,a) =\nabla_f(a,1)=\nabla_f(a,a+1)=4$ if and only if $a\in \F_4 \setminus \F_2$, and $\nabla_f(a,1)=\nabla_f(a,a+1)=0$ otherwise. To summarize, when $1\in \{a,b,a+b\}$, 
\begin{equation*}
\nabla_f(a,b)=
\begin{cases}
4,&\text{  if }a,b\in\F_4^*,\\
0,   &\text{ otherwise.}
\end{cases}
\end{equation*}

Now we assume that $ab(a+b)\ne 0$ and $1\not\in \{a,b,a+b\}$.
If $x\in \{0,a,b,a+b\}$, then \eqref{FBU_eqn} is equivalent to 
\begin{equation}\label{Inv01_char_even_proof_eq1}
a^2b+ab^2+a^2+b^2+ab=0.
\end{equation} 
If $x\in \{1,a+1,b+1,a+b+1\}$, then \eqref{FBU_eqn} is equivalent to 
\begin{equation}\label{Inv01_char_even_proof_eq2}
a^2+b^2+ab=1.
\end{equation}
If $x\not \in \{0,a,b,a+b\}\cup\{1,a+1,b+1,a+b+1\}$, then \eqref{FBU_eqn} is equivalent to $x^{2^n-2}+(x+a)^{2^n-2}+(x+b)^{2^n-2}+(x+a+b)^{2^n-2}=0$, which is equivalent to $ab(a+b)=0$, a contradiction. Hence, \eqref{FBU_eqn} has at most $8$ solutions $0,a,b,a+b,1,a+1,b+1,a+b+1$.

Next we show that $\nabla_f(a,b)=8$ if and only if $a^3+a+1=b^3+b+1=0$. Assume that $\nabla_f(a,b)=8$. Then, we have that  \eqref{Inv01_char_even_proof_eq1} and \eqref{Inv01_char_even_proof_eq2} hold. Adding \eqref{Inv01_char_even_proof_eq1} and \eqref{Inv01_char_even_proof_eq2}, we have
\begin{equation}\label{Inv01_char_even_proof_eq3}
a^2b+ab^2=1.
\end{equation}
Dividing $a^2$ on both sides of \eqref{Inv01_char_even_proof_eq2} and dividing $a^3$ on both sides of \eqref{Inv01_char_even_proof_eq3}, we have 
\begin{equation*}
\left(\Frac{b}{a}\right)^2+\Frac{b}{a}=1+\left(\Frac{1}{a}\right)^2\text{ and }\left(\Frac{b}{a}\right)^2+\Frac{b}{a}=\left(\Frac{1}{a}\right)^3
\end{equation*} 
respectively, and hence adding the two above equations implies that $a^3+a+1=0$. By a symmetric method that dividing $b^2$ to \eqref{Inv01_char_even_proof_eq2} and dividing $b^3$ to \eqref{Inv01_char_even_proof_eq3}, we also have $b^3+b+1=0$.

Conversely, we assume that $a^3+a+1=b^3+b+1=0$ with $a\ne b$, and it is enough to show that \eqref{Inv01_char_even_proof_eq1} and \eqref{Inv01_char_even_proof_eq2} hold. We denote another solution of $x^3+x+1=0$ by $\alpha$ so that $\alpha=a+b$ and $\frac{1}{\alpha}=ab$. Then we have
\begin{align*}
a^2b+ba^2+a^2+b^2+ab&=ab(a+b)+(a+b)^2+ab=1+\alpha^2+\frac{1}{\alpha}=\frac{\alpha^3+\alpha+1}{\alpha}=0,\\
a^2+b^2+ab&=(a+b)^2+ab=\alpha^2+\frac{1}{\alpha}=\frac{\alpha^3+1}{\alpha}=\frac{\alpha}{\alpha}=1,
\end{align*}
which completes the proof.
\end{proof}

Next, we study the second-order zero differential spectrum of $f=Inv \circ (0,1)$. Since there are three elements $x\in \Fbn$ such that $x^3+x+1=0$ if and only if $3\mid n$, we have the following :

\begin{equation}\label{Inv01_DS8}
\omega_8 =
\begin{cases}
6, &\text{ if }3\mid n,\\
0, &\text{ if }3\nmid n.
\end{cases}
\end{equation}
We denote 
\begin{align*}
S_1 &= \{(a,b)\in \Fbnmul\times \Fbnmul : \eqref{Inv01_char_even_proof_eq1}\text{ holds with }a\ne b\}\\
S_2 &= \{(a,b)\in \Fbnmul\times \Fbnmul : \eqref{Inv01_char_even_proof_eq2}\text{ holds with }a\ne b\}\\
S_3 &= \{(a,b)\in \Fbnmul\times \Fbnmul : (a,b)\in \F_4^* \times \F_4^*, a\ne b\}
\end{align*}

Then, by Theorem \ref{Inv01_char_even_thm}, we have $\omega_8 = \#(S_1 \cap S_2)$ and
\begin{equation}\label{Inv01_DS4_eqn}
\omega_4 = \#[(S_1 \cup S_2) -(S_1 \cap S_2)] +\#S_3=\#S_1 + \#S_2+\#S_3-2\omega_8,
\end{equation}
because it can be easily checked that $S_3 \cap(S_1 \cup S_2) =\emptyset$. Note that \eqref{Inv01_char_even_proof_eq1} is equivalent to $(a+1)b^2+a(a+1)b+a^2=0$. So, for a fixed $a\in \Fbn\setminus\{0,1\}$, there are two elements $b\in \Fbn$ satisfying \eqref{Inv01_char_even_proof_eq1} if and only if $\tr\left(\frac{1}{a+1}\right)=0$, by Lemma \ref{quad lemma}. Since $\{\frac{1}{a+1} : a\in \Fbn\setminus\{0,1\}\}=\Fbn\setminus\{0,1\}$, we have
\begin{equation*}
\#S_1 = 
\begin{cases}
2^{n}-2, &\text{if }2\nmid n,\\
2^{n}-4, &\text{if }2\mid n.
\end{cases}
\end{equation*}
Similarly, \eqref{Inv01_char_even_proof_eq2} is equivalent to $b^2 +ab +a^2+1=0$. So, for a fixed $a\in \Fbnmul$, there are two elements $b\in \Fbn$ satisfying \eqref{Inv01_char_even_proof_eq2} if and only if $\tr\left(\frac{1}{a}+1\right)=0$, by Lemma \ref{quad lemma}. $(a,b)=(1,1)$ is a solution of \eqref{Inv01_char_even_proof_eq2} with $ab(a+b)=0$ and then we have $\nabla_f(1,1)=2^n$, so the case $a=1$ need to be excluded. We have $\{\frac{1}{a}+1 : a\in \Fbn\setminus\{0,1\}\}=\Fbn\setminus\{0,1\}$ and hence $\#S_2 = \#S_1$.
Furthermore, $\#S_3=6$ when $n$ is even, and $\#S_3=0$ when $n$ is odd. Therefore, by \eqref{Inv01_DS4_eqn}, we have
\begin{equation}\label{Inv01_DS4}
\omega_4 =
\begin{cases}
2^{n+1}-14, &\text{ if }2\mid n, 3\mid n,\\
2^{n+1}-2, &\text{ if }2\mid n, 3\nmid n,\\
2^{n+1}-16, &\text{ if }2\nmid n, 3\mid n,\\
2^{n+1}-4, &\text{ if }2\nmid n, 3\nmid n.
\end{cases}
\end{equation}
Observe that $\#\{(a,b)\in \Fbn \times \Fbn : ab(a+b)\ne 0\}=(2^n-2)(2^n-1)$. By Theorem \ref{Inv01_char_even_thm}, we have
\begin{equation*}\label{Inv01_DS0}
\omega_0 = (2^n-2)(2^n-1)-\omega_4-\omega_8=
\begin{cases}
2^{2n}-5\cdot 2^{n}+10, &\text{ if }2\mid n, 3\mid n,\\
2^{2n}-5\cdot 2^{n}+4, &\text{ if }2\mid n, 3\nmid n,\\
2^{2n}-5\cdot 2^{n}+12, &\text{ if }2\nmid n, 3\mid n,\\
2^{2n}-5\cdot 2^{n}+6, &\text{ if }2\nmid n, 3\nmid n.
\end{cases}
\end{equation*}

From \eqref{Inv01_DS8} and \eqref{Inv01_DS4}, we have the following corollary on the second-order zero differential uniformity of $Inv \circ (0,1)$.

\begin{corollary} Let $p=2$ and $f(x) = Inv\circ(0,1)$. Then,
\begin{equation*}
\nabla_f=
\begin{cases}
4, &\text{ if }3\nmid n,\\
8, &\text{ if }3\mid n.
\end{cases}
\end{equation*}
\end{corollary}

\subsection{The second-order zero differential uniformity in odd characteristic}\label{subsec_inv01_FBU}

In the next theorem, we study the second-order zero differential uniformity of $Inv \circ (0,1)$ in odd characteristic.

\begin{theorem}\label{Inv01_char_odd_thm} Let $p$ be an odd prime and $f(x) = Inv\circ(0,1)$ on $\Fpn$. Then, $\nabla_f=4$ if $p=29$ or $p=37$, and $\nabla_f\le 3$ otherwise. Furthermore, $\nabla_f(a,b)=4$ if and only if one of the following conditions holds.
\begin{enumerate}
\item $p=29$ and $(a,b)=(\pm 2, \pm 12)$ or $(a,b)=(\pm12, \pm 2)$.
\item $p=37$ and $(a,b)=(\pm 2, \pm 6)$ or $(a,b)=(\pm6, \pm 2)$.
\end{enumerate}
\end{theorem}

\begin{proof}
We first consider the case that $a, b =\pm 1$. If $a=1$ then \eqref{FBU_eqn} is equivalent to 
\begin{equation}\label{podd_a1_eqn}
f(x+b+1)-f(x+1)-f(x+b)+f(x)=0.
\end{equation}
If $a=b=1$, then \eqref{podd_a1_eqn} is equivalent to 
\begin{equation}\label{podd_a1b1_eqn}
f(x+2)-2f(x+1)+f(x)=0.
\end{equation} 
\begin{itemize}
\item If $x=0$, then \eqref{podd_a1b1_eqn} is equivalent to $2^{-1}=-1$, which implies $p=3$.
\item If $x=-1$, then \eqref{podd_a1b1_eqn} is equivalent to $-3=0$, which implies $p=3$.
\item If $x=-2$, then \eqref{podd_a1b1_eqn} is equivalent to $2^{-1}=3$, which implies $p=5$.
\item If $x=1$, then \eqref{podd_a1b1_eqn} holds when $p=3$, and we have $3=1$, which is a contradiction, when $p>3$. 
\item If $x\not \in \{0,\pm 1, -2\}$, then \eqref{podd_ab1_eqn} is equivalent to $2=0$, which is a contradiction.
\end{itemize}
Hence, we have 
\begin{equation}\label{podd_nabla_f11}
\nabla_f(1,1)=
\begin{cases}
3, &\text{ if }p=3,\\
1, &\text{ if }p=5,\\
0, &\text{ otherwise.}
\end{cases}
\end{equation}
By \eqref{-a-b_eqn}, we also have $\nabla_f(1,-1)=\nabla_f(1,1)$ is expressed as \eqref{podd_nabla_f11}. Now we assume that $b\ne \pm 1$.
\begin{equation*}
\begin{array}{ll}
x=0\ : \eqref{podd_a1_eqn}\ \Leftrightarrow\  b^2+b-1=0, & x=1\ : \eqref{podd_a1_eqn}\ \Leftrightarrow\  b^2+3b+4=0,\\
x=-1\ : \eqref{podd_a1_eqn}\ \Leftrightarrow\  2b^2-2b+1=0, & x=-b\ : \eqref{podd_a1_eqn}\ \Leftrightarrow\  b^2-b-1=0,\\
x=-b+1\ : \eqref{podd_a1_eqn}\ \Leftrightarrow\  b^2-3b+4=0, & x=-b-1\ : \eqref{podd_a1_eqn}\ \Leftrightarrow\  2b^2+2b+1=0.
\end{array}
\end{equation*}
 If $x\not \in \{0,\pm 1, -b, -b\pm 1\}$, then \eqref{podd_a1_eqn} is equivalent to $b(2x+b+1)=0$, and hence we have that $x=-\frac{b+1}{2}$ is a solution of \eqref{podd_a1_eqn}. 

We consider the cases that \eqref{podd_a1_eqn} has at least two solutions in $\{0,\pm 1, -b, -b\pm1\}$. We can see that $x=\alpha$ and $x=-b+\alpha$ cannot be solutions of \eqref{podd_a1_eqn} simultaneously, where $\alpha \in\{0,\pm1\}$. So, we do not consider those cases below.
\begin{itemize}
\item If $x=0$ and $x=1$ are solutions of \eqref{podd_a1_eqn}, then we have $b^2+b-1=b^2+3b+4=0$ which implies $b=-\frac 5 2$. If $b=-\frac 5 2$, then $b^2+b-1=b^2+3b+4=\frac{11}{4}$. So, we have $p=11$, and $b=-\frac 5 2 = 3$. 
\item If $x=0$ and $x=-1$ are solutions of \eqref{podd_a1_eqn}, then we have $b^2+b-1=2b^2-2b+1=0$ which implies $b=\frac3 4$.  If $b=\frac 3 4$, then $b^2+b-1= \frac{5}{16}$ and $2b^2-2b+1=\frac58$. So, we have $p=5$, and $b=\frac 3 4 = 2$. 
\item If $x=0$ and $x=-b+1$ are solutions of \eqref{podd_a1_eqn}, then we have $b^2+b-1=b^2-3b+4=0$ which implies $b=\frac5 4$.  If $b=\frac 5 4$, then $b^2+b-1= b^2-3b+4=\frac{29}{16}$. So, we have $p=29$, and $b=\frac 5 4 = -6$.
\item If $x=0$ and $x=-b-1$ are solutions of \eqref{podd_a1_eqn}, then we have $b^2+b-1=2b^2+2b+1=0$ which implies $3=0$ and hence we have $p=3$. Note that $b\in \F_{3^n}$ with $b^2+b-1=0$ exists if and only if $n$ is even.
\item If $x=1$ and $x=-1$ are solutions of \eqref{podd_a1_eqn}, then we have $b^2+3b+4=2b^2-2b+1=0$ which implies $b=-\frac7 8$.  If $b=-\frac 7 8$, then $b^2+3b+4=\frac{137}{64}$ and $2b^2-2b+1=\frac{137}{29}$. So, we have $p=137$, and $b=-\frac 7 8 = 119$.
\item If $x=1$ and $x=-b$ are solutions of \eqref{podd_a1_eqn}, then we have $b^2+3b+4=b^2-b-1=0$ which implies $b=-\frac5 4$.  If $b=-\frac 5 4$, then $b^2+3b+4=b^2-b-1=\frac{29}{16}$. So, we have $p=29$, and $b=-\frac 5 4 = 6$.
\item If $x=1$ and $x=-b-1$ are solutions of \eqref{podd_a1_eqn}, then we have $b^2+3b+4=2b^2+2b+1=0$ which implies $b=-\frac7 4$.  If $b=-\frac 7 4$, then $2b^2+2b+1=\frac{29}8$ and $b^2+3b+4=\frac{29}{16}$. So, we have $p=29$, and $b=-\frac 7 4 = 20$.
\item If $x=-1$ and $x=-b$ are solutions of \eqref{podd_a1_eqn}, then we have $2b^2-2b+1=b^2-b-1=0$ which implies $3=0$ and hence we have $p=3$. Note that $b\in \F_{3^n}$ with $b^2-b-1=0$ exists if and only if $n$ is even.
\item If $x=-1$ and $x=-b+1$ are solutions of \eqref{podd_a1_eqn}, then we have $2b^2-2b+1=b^2-3b+4=0$ which implies $b=\frac7 4$.  If $b=\frac 7 4$, then $2b^2-2b+1=\frac{29}8$ and $b^2-3b+4=\frac{29}{16}$. So, we have $p=29$, and $b=\frac 7 4 = 9$.
\item If $x=-b$ and $x=-b+1$ are solutions of \eqref{podd_a1_eqn}, then we have $b^2-b-1=b^2-3b+4=0$ which implies $b=\frac5 2$.  If $b=\frac 3 4$, then $b^2-b-1=b^2-3b+4= \frac{11}{4}$. So, we have $p=11$, and $b=\frac 5 2 = 8$. 
\item If $x=-b$ and $x=-b-1$ are solutions of \eqref{podd_a1_eqn}, then we have $b^2-b-1=2b^2+2b+1=0$ which implies $b=-\frac3 4$.  If $b=-\frac 3 4$, then $b^2-b-1= \frac{5}{16}$ and $2b^2+2b+1=\frac58$. So, we have $p=5$, and $b=-\frac 3 4 = 3$.
\item If $x=-b+1$ and $x=-b-1$ are solutions of \eqref{podd_a1_eqn}, then we have $b^2-3b+4=2b^2+2b+1=0$ which implies $b=\frac7 8$.  If $b=\frac 7 8$, then $b^2-3b+4=\frac{137}{64}$ and $2b^2+2b+1=\frac{137}{29}$. So, we have $p=137$, and $b=\frac 7 8 = 18$.
\end{itemize}
There are no common case in above, and hence \eqref{podd_a1_eqn} has at most two solutions in $\{0,\pm 1, -b, -b\pm 1\}$ and at most one solution in $\Fpn \setminus\{0,\pm 1, -b, -b\pm 1\}$. Therefore, we have $\nabla_f(1,b)\le 3$ for all $b\in \Fpn$.\\

Now we assume that $a,b \not \in \{0,\pm 1\}$. Next we consider the case that $a\pm b \in \{0,\pm 1\}$. \\
(i) If $a+b=0$, then \eqref{FBU_eqn} is equivalent to 
\begin{equation}\label{podd_ab0_eqn}
2f(x)-f(x+a)-f(x-a)=0
\end{equation}
If $x=0$, $x=1$ and $x=-1$, then \eqref{podd_ab0_eqn} is equivalent to $2=0$, $(a+1)^{-1}=(a-1)^{-1}$ and $\frac{a^2}{(a-1)(a+1)}=0$, which are contradictions. Moreover,
\begin{equation*}
\begin{array}{ll}
x=-a\ : \eqref{podd_ab0_eqn}\ \Leftrightarrow\  a=-\frac 3 2, & x=-a+1\ : \eqref{podd_ab0_eqn}\ \Leftrightarrow\  a=3^{-1},\\
x=a\ : \eqref{podd_ab0_eqn}\ \Leftrightarrow\  a=\frac 3 2, & x=a+1\ : \eqref{podd_ab0_eqn}\ \Leftrightarrow\  a=-3^{-1}.
\end{array}
\end{equation*}
 If $x\not \in  \{0, \pm1, \pm a, \pm a+1\}$ then $0,1\not \in \{x,x+1,x+a,x+1-a\}$, and hence \eqref{podd_ab0_eqn} is equivalent to $\frac{a^2}{x(x+a)(x-a)}=0$, which contradicts the assumption.

\eqref{podd_ab0_eqn} has two solutions $x=-a+1$ and $x=a$ when $3^{-1}=\frac 3 2$ or equivalently $9=2$, which implies $p=7$, and we have $a=3^{-1}=5$. Similarly, when $p=7$ and $a=-3^{-1} = - \frac 3 2 =2$, \eqref{podd_ab0_eqn} has two solutions $x=-a$ and $x=a+1$. 
\eqref{podd_ab0_eqn} has two solutions $x=-a$ and $x=-a+1$ when $3^{-1}=-\frac 3 2$ or equivalently $-9=2$, which implies $p=11$, and we have $a=3^{-1}=4$. Similarly, when $p=11$ and$a=-3^{-1} =  \frac 3 2 =7$, \eqref{podd_ab0_eqn} has two solutions $x=-a$ and $x=-a+1$. Therefore, we have 
\begin{equation*}
\nabla_f(a,-a)=
\begin{cases}
2,&\text{ if }(p,a)\in \{(7,2), (7,5), (11,4),(11,7)\},\\
1,&\text{ if }(p,a)\not \in \{(7,2), (7,5), (11,4),(11,7)\}\text{ and }a\in \{\pm 3^{-1},\pm \frac 3 2\},\\
0,&\text{otherwise.}
\end{cases}
\end{equation*}
(ii) If $a+b=1$, then \eqref{FBU_eqn} is equivalent to 
\begin{equation}\label{podd_ab1_eqn}
f(x+1)-f(x+a)-f(x+1-a)+f(x)=0
\end{equation}
Then, we have
\begin{equation*}
\begin{array}{ll}
x=0\ : \eqref{podd_ab1_eqn}\ \Leftrightarrow\  a^2-a+1=0, & x=1\ : \eqref{podd_ab1_eqn}\ \Leftrightarrow\  a^2-a+4=0,\\
x=-a\ : \eqref{podd_ab1_eqn}\ \Leftrightarrow\  2a^3-2a+1=0, & x=-a+1\ : \eqref{podd_ab1_eqn}\ \Leftrightarrow\  3a^2-7a+4=(3a-4)(a-1)=0\ \Rightarrow \ a=\frac 4 3 ,\\
x=a\ : \eqref{podd_ab1_eqn}\ \Leftrightarrow\  a=-3^{-1}, & x=a-1\ : \eqref{podd_ab1_eqn}\ \Leftrightarrow\  2a^3-6a^2+4a-1=0.
\end{array}
\end{equation*}
If $x=-1$, then \eqref{podd_ab1_eqn} is equivalent to $1=0$, which is a contradiction. If $x\not \in  \{0, \pm1, \pm a, \pm(a-1)\}$ then $0,1\not \in \{x,x+1,x+a,x+1-a\}$, and hence \eqref{podd_ab1_eqn} is equivalent to $a(a-1)(2x+1)=0$, so \eqref{podd_ab1_eqn} has a solution $x=-\frac 1 2$. 

We consider the cases that \eqref{podd_ab1_eqn} has at least two solutions in $\{0,\pm 1, \pm a, \pm (a-1)\}$. 
\begin{itemize}
\item If $x=0$ and $x=1$  are solutions of \eqref{podd_ab1_eqn}, then we have $p=3$. But, if $p=3$, then $0=a^2-a+1=(a+1)^2$ implies $a=-1$, which contradicts $a \not \in \{0,\pm1\}$.
\item If $x=0$ and $x=-a$ are solutions of \eqref{podd_ab1_eqn}, then we have $a^2-a+1=0$ and $2a^3-2a+1=0$. So,  $0=2a^3-2a+1=(a^2-a+1)(2a+2)-2a-1=-2a-1$ and hence we have $a=-\frac 1 2$. Then $a^2-a+1=2a^3-2a+1=\frac 74$ and hence we have $p=7$ and then $a=-\frac 1 2=3$. 
\item If $x=0$ and $x=-a+1$ are solutions of \eqref{podd_ab1_eqn}, then $a=\frac 4 3$ and $0=a^2-a+1=\frac{13}{9}$. So, we have $p=13$, and then $a=\frac 4 3 = 10$.
\item If $x=0$ and $x=a$ are solutions of \eqref{podd_ab1_eqn}, then $a=-\frac 1 3$ and $0=a^2-a+1=\frac{13}{9}$. So, we have $p=13$, and then $a=-\frac 1 3 = 4$.
\item If $x=0$ and $x=a-1$ are solutions of \eqref{podd_ab1_eqn}, then we have $a^2-a+1=0$ and $2a^3-6a^2+4a-1=0$. So, $0=2a^3-6a^2+4a-1=(a^2-a+1)(2a-4)-2a+3=-2a+3$ and hence we have $a=\frac 3 2$. Then $a^2-a+1=\frac74$ and $2a^3-6a^2+4a-1=-\frac74$, and hence we have $p=7$ and $a=\frac 3 2=5$. 
\item  If $x=1$ and $x=-a$ are solutions of \eqref{podd_ab1_eqn}, then we have $a^2-a+4=0$ and $2a^3-2a+1=0$. So,  $0=2a^3-2a+1=(a^2-a+4)(2a+2)-8a-7=-8a-7$ and hence we have $a=-\frac  7 8$. Then $a^2-a+4=\frac{361}{64}$ and $2a^3-2a+1=\frac{361}{256}$, so we have $p=19$ and $a=-\frac 7 8 =11$.  
\item  If $x=1$ and $x=-a+1$ are solutions of \eqref{podd_ab1_eqn}, then we have $a=\frac 4 3$ and $0=a^2-a+4=\frac{40}{9}$. So, we have $p=5$, and then $a=\frac 43=3$. 
\item  If $x=1$ and $x=a$ are solutions of \eqref{podd_ab1_eqn}, then we have $a^2-a+4=0$ and $a=-\frac 1 3$ and $0=a^2-a+4=\frac{40}{9}$. So, we have $p=5$, and then $a=-\frac 13=3$. 

\item If $x=1$ and $x=a-1$ are solutions of \eqref{podd_ab1_eqn}, then we have $a^2-a+4=0$ and $2a^3-6a^2+4a-1=0$. So, $0=2a^3-6a^2+4a-1=(a^2-a+4)(2a-4)-8a+15=-8a+15$ and hence we have $a=\frac {15} 8$. Then $a^2-a+4=\frac{361}{64}$ and $2a^3-6a^2+4a-1=-\frac{361}{256}$, so we have $p=19$ and $a=\frac {15} 8=9$.
\item If $x=-a$ and $x=-a+1$ are solutions of \eqref{podd_ab1_eqn}, then we have $a=\frac 4 3$ and $0=2a^3-2a+1=\frac{83}{27}$. So, we have $p=83$, and then $a=\frac 43=29$. 
\item If $x=-a$ and $x=a$ are solutions of \eqref{podd_ab1_eqn}, then we have $a=-\frac 1 3$ and $0=2a^3-2a+1=\frac{43}{27}$. So, we have $p=43$, and then $a=-\frac 43=14$. 
\item If $x=-a$ and $x=a-1$ are solutions of \eqref{podd_ab1_eqn}, then we have $2a^3-2a+1=0$ and $2a^3-6a^2+4a-1=0$. Then, we have $0=(2a^3-2a+1)(1-a)+(2a^3-6a^2+4a-1)(1+a)-2a^3$, which contradicts $a\ne 0$. 
\item If $x=-a+1$ and $x=a$ are solutions of \eqref{podd_ab1_eqn}, then we have $a=\frac 43$ and $a=-\frac13$. Then, we have $p=5$ and $a=-\frac13=3$. 
\item If $x=-a+1$ and $x=a-1$ are solutions of \eqref{podd_ab1_eqn}, then we have $a=\frac 43$ and $0=2a^3-6a^2+4a-1=-\frac{43}{27}$. Then, we have $p=43$ and $a=\frac 4 3 = 30$.
\item If $x=a$ and $x=a-1$ are solutions of \eqref{podd_ab1_eqn}, then we have $a=-\frac 13$ and $0=2a^3-6a^2+4a-1=-\frac{83}{27}$. Then, we have $p=83$ and $a=-\frac 1 3 = 55$.
\end{itemize}
In the cases considered above, the only candidate that \eqref{podd_ab1_eqn} has more than two solutions in $\{0,\pm 1, \pm a, \pm (a-1)\}$ is when $p=5$ and $a=3$. In this case, $x=1$, $x=-a+1$ and $x=a$ are solutions of \eqref{podd_ab1_eqn}, but we have $-a+1=3=a$, and hence \eqref{podd_ab1_eqn} has only two solutions in $\{0,\pm 1, \pm a, \pm (a-1)\}$. 
Therefore, \eqref{podd_ab1_eqn} has at most two solutions in $\{0,\pm 1, \pm a, \pm (a-1)\}$ and at most one solution $\Fpn \setminus \{0,\pm 1, \pm a, \pm (a-1)\}$, and hence $\nabla_f(a,1-a)\le 3$ for all $a\in \Fpn$. 

The other cases with $a\pm b \in \{0,\pm 1\}$ are covered by the above cases (i) and (ii).
\begin{itemize}
\item If $a+b=-1$, then $-a-b=1$. By \eqref{-a-b_eqn} and (ii), we have $\nabla_f(a,b)=\nabla_f(-a,-b)\le 3$.
\item If $a-b=0$, then $b=a$. By \eqref{-a-b_eqn} and (i), we have $\nabla_f(a,a)=\nabla_f(a,-a)\le 2$.
\item If $a-b=1$, then $b=a+1$. By \eqref{-a-b_eqn} and the above case that $a+b=-1$, we have $\nabla_f(a,a+1)=\nabla_f(a,-a-1)\le 3$.
\item If $a-b=-1$, then $b=a-1$. By \eqref{-a-b_eqn} and (ii), we have $\nabla_f(a,a-1)=\nabla_f(a,-a+1)\le 3$.
\end{itemize}

Next we consider that $a,b,a\pm b \not \in \{0,\pm1\}$.
\begin{align}
x=0 &: \eqref{FBU_eqn} \ \Leftrightarrow\  a^2+ab+b^2=ab(a+b). \label{podd x=0 eqn}\\
x=1 &: \eqref{FBU_eqn} \ \Leftrightarrow\ a^2+ab+b^2=-2(a+b)-1. \label{podd x=1 eqn}\\
x=-a &: \eqref{FBU_eqn} \ \Leftrightarrow\ a^2-ab+b^2=ab(a-b). \label{podd x=-a eqn}\\
x=-a+1 &: \eqref{FBU_eqn} \ \Leftrightarrow\ a^2-ab+b^2=2(a-b)-1.  \label{podd x=-a+1 eqn}\\
x=-b &: \eqref{FBU_eqn} \ \Leftrightarrow\ a^2-ab+b^2=-ab(a-b). \label{podd x=-b eqn}\\
x=-b+1 &: \eqref{FBU_eqn} \ \Leftrightarrow\ a^2-ab+b^2=-2(a-b)-1. \label{podd x=-b+1 eqn}\\ 
x=-a-b &: \eqref{FBU_eqn} \ \Leftrightarrow\ a^2+ab+b^2=-ab(a+b). \label{podd x=-a-b eqn}\\
x=-a-b+1 &: \eqref{FBU_eqn} \ \Leftrightarrow\ a^2+ab+b^2=2(a+b)-1.  \label{podd x=-a-b+1 eqn}
\end{align}
Any two of equations among \eqref{podd x=0 eqn}, \eqref{podd x=-a eqn}, \eqref{podd x=-b eqn}, and \eqref{podd x=-a-b eqn} cannot hold simultaneously. For example, if $x=0$ and $x=-a$ are two solutions of \eqref{FBU_eqn}, then we have $ab(1-b)=0$ from \eqref{podd x=0 eqn} and \eqref{podd x=-a eqn}, which contradicts the assumption of the case that $a,b\not \in \{0,\pm1\}$. Hence \eqref{FBU_eqn} has at most one solution in $\{0,-a,-b,-a-b\}$. Hence, we obtain that \eqref{FBU_eqn} cannot have two solutions in $\{0,-a,-b,-a-b\}$. \\
We consider any two of \eqref{podd x=1 eqn}, \eqref{podd x=-a+1 eqn}, \eqref{podd x=-b+1 eqn} and \eqref{podd x=-a-b+1 eqn} hold simultaneously. 
\begin{itemize}
\item If \eqref{podd x=1 eqn} and \eqref{podd x=-a+1 eqn} holds, then we have $b=-2$ and then we have $a^2=-1$.
\item If \eqref{podd x=1 eqn} and \eqref{podd x=-b+1 eqn} holds, then we have $a=-2$ and then we have $b^2=-1$.
\item If \eqref{podd x=1 eqn} and \eqref{podd x=-a-b+1 eqn} holds, then we have $a+b=0$, which contradicts $a+ b \not \in \{0,\pm1\}$.
\item If \eqref{podd x=-a+1 eqn} and \eqref{podd x=-b+1 eqn} holds, then we have $a-b=0$, which contradicts $a- b \not \in \{0,\pm1\}$.
\item If \eqref{podd x=-a+1 eqn} and \eqref{podd x=-a-b+1 eqn} holds, then we have $a=2$ and then we have $b^2=-1$.
\item If \eqref{podd x=-a+1 eqn} and \eqref{podd x=-a-b+1 eqn} holds, then we have $b=2$ and then we have $a^2=-1$.
\end{itemize}

We can see that \eqref{FBU_eqn} has at most two solutions in $\{1,-a+1,-b+1,-a-b+1\}$, because we have $a,b=\pm2$ and hence $a+b=0$ or $a-b=0$, which contradict to the assumption that $a,b,a\pm b \not \in \{0,\pm1\}$.
We substitute the above cases to \eqref{podd x=0 eqn}, \eqref{podd x=-a eqn}, \eqref{podd x=-b eqn}, and \eqref{podd x=-a-b eqn}. We only consider the cases $a=\pm 2$ and $b^2=-1$, because the cases with $b=\pm2$ and $a^2=-1$ are similar.
\begin{itemize}
\item We substitute $a=2$, $b^2=-1$ to \eqref{podd x=0 eqn}, then we have $b=\frac 5 2$. Then $b^2=\frac{25}{4}=-1$ implies $p=29$, and then $b=\frac 5 2=17=-12$.
\item We substitute $a=-2$, $b^2=-1$ to \eqref{podd x=0 eqn}, then we have $b=\frac 1 6$. Then $b^2=\frac{1}{36}=-1$ implies $p=37$, and then $b=\frac 1 6=3-6$.
\end{itemize}
Similarly, we obtain the followings:
\begin{align*}
a=2, b^2=-1\text{ in }\eqref{podd x=-a eqn}&\  \Rightarrow\ p=37\text{ and }b=\frac 1 6 =-6.\\ 
a=-2, b^2=-1\text{ in }\eqref{podd x=-a eqn}&\  \Rightarrow\ p=29\text{ and }b=\frac 5 2  =-12.\\
a=2, b^2=-1\text{ in }\eqref{podd x=-b eqn}&\  \Rightarrow\ p=29\text{ and }b=-\frac 5 2  =12.\\
a=-2, b^2=-1\text{ in }\eqref{podd x=-b eqn}&\  \Rightarrow\ p=37\text{ and }b=-\frac 1 6 =6.\\
a=2, b^2=-1\text{ in }\eqref{podd x=-a-b eqn}&\  \Rightarrow\ p=37\text{ and }b=-\frac 1 6 =6.\\ 
a=-2, b^2=-1\text{ in }\eqref{podd x=-a-b eqn}&\  \Rightarrow\ p=29\text{ and }b=-\frac 5 2  =12.
\end{align*}
To summarize, \eqref{FBU_eqn} has three solutions in $\{0,1,-a,-a+1,-b,-b+1,-a-b,-a-b+1\}$ if $(p,a,b)\in \{(29,\pm 2, \pm 12), (37, \pm 2, \pm 6)\}$. Similarly, for the cases that $b=\pm2$ and $a^2=-1$, we can see that \eqref{FBU_eqn} has three solutions in $\{0,1,-a,-a+1,-b,-b+1,-a-b,-a-b+1\}$ if $(p,a,b)\in \{(29,\pm 12, \pm 2), (37, \pm 6, \pm 2)\}$.

If $x\not \in  \{0,1,-a,-a+1,-b,-b+1,-(a+b),-(a+b)+1\}$, then $0,1\not \in \{x,x+a,x+b,x+a+b\}$, and hence \eqref{FBU_eqn} is equivalent to $ab(2x+a+b)=0$, hence $x=-\frac{a+b}{2}$ is a solution of \eqref{FBU_eqn}. We can easily check that $-\frac{a+b}2\not\in \{0,1,-a,-a+1,-b,-b+1,-(a+b),-(a+b)+1\}$ for all 
\begin{equation}\label{podd_FBU4}
(p,a,b)\in \{(29,\pm 2, \pm 12),(29,\pm 12, \pm 2), (37, \pm 2, \pm 6),(37, \pm 6, \pm 2)\},
\end{equation}
and hence we have $\nabla_f(a,b)=4$ if and only if \eqref{podd_FBU4} holds, which completes the proof.
\end{proof}

\begin{remark} 
By \eqref{podd_nabla_f11} and Theorem \ref{Inv01_char_odd_thm}, we have $\nabla_f=3$ when $p=3$.
\end{remark}

\section{The second-order zero differential uniformity of $Inv \circ (1,\gamma)$}\label{sec_gamnot0}

\subsection{The second-order zero differential spectrum in even characteristic}
 
\begin{theorem}\label{Inv1r_char_even_thm} Let $p=2$ and $f(x) = Inv\circ(1, \gamma)$ where $\gamma\in \Fbn\setminus\{0,1\}$. If $ab(a+b)\ne 0$, then we have
\begin{equation*}
\nabla_f(a,b)=
\begin{cases}
8, &\text{if }1\in\{a,b,a+b\}, (\{a,b\}\setminus\{1\})\subset S_{\gamma,1},\\
&\gamma\in\{a,b,a+b\}, (\{a,b\}\setminus\{\gamma\})\subset S_{\gamma,2},\\
&\frac b a \in\F_4\setminus \F_2\text{ and }a^3\in \{\gamma+1,\gamma^3+\gamma^2\}, \{a,b,a+b\}\cap\{1,\gamma\}=\emptyset\\
4,&\text{if }\gamma \in \F_4\setminus\F_2, a,b\in\F_4^*, a\ne b,\\
&1\in\{a,b,a+b\}, (\{a,b\}\setminus\{1\})\subset S_{\gamma,3}\cup S_{\gamma,4}\\
&\gamma\in\{a,b,a+b\}, (\{a,b\}\setminus\{\gamma\})\subset S_{\gamma,5}\cup S_{\gamma,6}\\
&\frac b a \not\in \F_4, (a,b)\in S_{\gamma,7}\cup S_{\gamma,8}, \{a,b,a+b\}\cap\{1,\gamma\}=\emptyset,
\\
&\frac b a \in \F_4\setminus\F_2, (a,b)\not\in S_{\gamma,7}\cup S_{\gamma,8}, \{a,b,a+b\}\cap\{1,\gamma\}=\emptyset,\\
0,& \text{otherwise.}
\end{cases}
\end{equation*}
where
\begin{equation*}
\begin{array}{ll}
S_{\gamma, 1}=\{c\in \F_8 \setminus \F_2 : \gamma = c(c+1)\},  & S_{\gamma, 2}=\{c\in \F_8 \setminus \F_2 : \gamma = c^2/(c+1)\},\\
S_{\gamma,3}=\{c\in\Fbn\setminus \F_8 : \gamma=c(c+1)\}, & S_{\gamma, 4}=\{c\in \Fbn \setminus \F_8 : \gamma^3+\gamma^2+(c^2+c+1)\gamma=1\},\\
S_{\gamma, 5}=\{c\in \Fbn \setminus \F_8 : \gamma = c^2/(c+1)\}, & S_{\gamma, 6}=\{c\in \Fbn \setminus \F_8 : \gamma^3+\gamma^2+(c+1)\gamma+c^2=1\},\\
S_{\gamma,7}=\{(a,b) \in \Fbnmul\times \Fbnmul : G_{\gamma,a,b}=\gamma+1\}, & S_{\gamma,8}=\{(a,b) \in \Fbnmul\times \Fbnmul : G_{\gamma,a,b}=\gamma^2(\gamma+1)\}\\
\end{array}
\end{equation*}
and $G_{\gamma,a,b}=ab(a+b)+(a^2+ab+b^2)(\gamma+1)$.
\end{theorem}

\begin{proof}  Our proof is divided into three cases, $1 \in \{a,b,a+b\}$, $\gamma \in \{a,b,a+b\}$ and $1,\gamma \not \in \{a,b,a+b\}$.

\noindent (Case 1) First, we consider the case $1 \in \{a,b,a+b\}$. Let $a=1$ and $b\ne 0,1$. Then, \eqref{FBU_eqn} is equivalent to 
\begin{equation}\label{Inv1r_char2_a1_eqn}
f(x)+f(x+1)+f(x+b)+f(x+b+1)=0.
\end{equation}
If $b\in\{\gamma,\gamma+1\}$, then \eqref{Inv1r_char2_a1_eqn} is equivalent to $f(x)+f(x+1)+f(x+\gamma)+f(x+\gamma+1)=0$.
If $x\in \{0,1,\gamma,\gamma+1\}$, then $f(x)+f(x+1)+f(x+\gamma)+f(x+\gamma+1)=0$ is equivalent to $\gamma^2+\gamma+1=0$. 
If $x\not \in \{0,1,\gamma,\gamma+1\}$, then $f(x)+f(x+1)+f(x+\gamma)+f(x+\gamma+1)=0$ is equivalent to $\gamma^2+\gamma=0$, which contradicts the assumption that $\gamma\not\in \{0,1\}$. \\
Assume that $b\not\in\{\gamma,\gamma+1\}$. 
If $x\in \{0,1,b,b+1\}$, then \eqref{Inv1r_char2_a1_eqn} is equivalent to $\gamma=b(b+1)$. 
If $x\in \{\gamma,\gamma+1,\gamma+b,\gamma+b+1\}$, then \eqref{Inv1r_char2_a1_eqn} is equivalent to $\gamma(\gamma^3+\gamma^2+(b^2+b+1)\gamma+1)=0$. 
If $x\not\in \{0,1,b,b+1\}\cup\{\gamma,\gamma+1,\gamma+b,\gamma+b+1\}$, \eqref{Inv1r_char2_a1_eqn} is equivalent to $b(b+1)=0$, which contradicts $b\ne0,1$. \eqref{Inv1r_char2_a1_eqn} has 8 solutions $\{0,1,b,b+1\}\cup\{\gamma,\gamma+1,\gamma+b,\gamma+b+1\}$ if and only if $\gamma=b(b+1)$ and $\gamma^3+\gamma^2+(b^2+b+1)\gamma+1=0$ hold, which implies that $b^6 + b^5 + b^4 + b^3 + b^2 + b + 1=0$ which is equivalent to $b\in \F_8 \setminus \F_2$. So, we have the following:
\begin{equation*}
\nabla_f(1,b)=
\begin{cases}
8, &\text{ if }\gamma=b(b+1)\text{ where }b\in \F_8 \setminus \F_2\\
4, &\text{ if }\gamma=b(b+1)\text{ or }\gamma^3+\gamma^2+(b^2+b+1)\gamma=1\text{ where }b\not\in \F_8\\
&\text{ or }b,\gamma\in\F_4 \setminus \F_2\\
0, &\text{ otherwise.}
\end{cases}
\end{equation*}
Note that the cases $b=1$ and $a+b=1$ can be similarly analyzed with the case $a=1$, and we summarize the cases $1\in \{a,b,a+b\}$ as follows :
\begin{equation*}
\nabla_f(a,b)=
\begin{cases}
8, &\text{ if }(\{a,b\}\setminus \{1\}) \subset S_{\gamma,1},\\
4, &\text{ if } (\{a,b\}\setminus\{1\})\subset S_{\gamma,3}\cup S_{\gamma,4}\\
&\text{ or }b,\gamma\in\F_4\setminus \F_2\\
0, &\text{ otherwise.}
\end{cases}
\end{equation*}

\noindent (Case 2)  Next, we consider the case $\gamma \in \{a,b,a+b\}$. Let $a=\gamma$ and $b\ne 0,\gamma$. Then, \eqref{FBU_eqn} is equivalent to 
\begin{equation}\label{Inv1r_char2_ar_eqn}
f(x)+f(x+\gamma)+f(x+b)+f(x+b+\gamma)=0.
\end{equation}
If $b\in\{1,\gamma+1\}$, then \eqref{Inv1r_char2_ar_eqn} is equivalent to $f(x)+f(x+1)+f(x+\gamma)+f(x+\gamma+1)=0$. Similarly with the case $a=1$, $f(x)+f(x+1)+f(x+\gamma)+f(x+\gamma+1)=0$ has 4 solutions $0,1,\gamma$ and $\gamma+1$ if $\gamma\in \F_4\setminus \F_2$, and no solution otherwise. \\
Assume that $b\not\in\{1,\gamma+1\}$. 
If $x\in \{0,\gamma,b,b+\gamma\}$, then \eqref{Inv1r_char2_ar_eqn} is equivalent to $(b+1)\gamma+b^2=0$. 
If $x\in \{1,\gamma+1,b+1,b+\gamma+1\}$, then \eqref{Inv1r_char2_ar_eqn} is equivalent to $\gamma^3+\gamma^2+(b+1)\gamma+b^2+1=0$.
If $x\not \in \{0,\gamma,b,b+\gamma\}\cup\{1,\gamma+1,b+1,b+\gamma+1\}$, then \eqref{Inv1r_char2_ar_eqn} is equivalent to $b\gamma(\gamma+b)=0$, which contradicts the assumption that $b\ne 0,\gamma$. \eqref{Inv1r_char2_ar_eqn} has 8 solutions $\{0,\gamma,b,b+\gamma\}\cup\{1,\gamma+1,b+1,b+\gamma+1\}$ if and only if $\gamma=b^2/(b+1)$ and $\gamma^3+\gamma^2+(b+1)\gamma+b^2+1=0$ which implies that $b^6 + b^5 + b^4 + b^3 + b^2 + b + 1=0$ which is equivalent to $b\in \F_8 \setminus \F_2$. Similarly with (Case 1), the cases $b=\gamma$ and $a+b=\gamma$ can be similarly analyzed with the case $a=\gamma$. So, we have the following:
\begin{equation*}
\nabla_f(\gamma,b)=
\begin{cases}
8, &\text{ if }(\{a,b\}\setminus \{\gamma\}) \subset S_{\gamma,2},\\
4, &\text{ if } (\{a,b\}\setminus\{\gamma\})\subset S_{\gamma,5}\cup S_{\gamma,6}\\
&\text{ or }b,\gamma\in\F_4\setminus \F_2\\
0, &\text{ otherwise.}
\end{cases}
\end{equation*}

\noindent (Case 3) Now we assume that $1,\gamma \not \in \{a,b,a+b\}$. 
\begin{align}
x\in \{0,a,b,a+b\} &: \eqref{FBU_eqn} \ \Leftrightarrow\  a^2+ab+b^2=0. \label{Inv1r_char2 x=0 eqn}\\
x\in \{1,a+1,b+1,a+b+1\} &: \eqref{FBU_eqn} \ \Leftrightarrow\ (a^2+ab+b^2+1)\gamma+1=a^2b+ab^2+a^2+ab+b^2. \label{Inv1r_char2 x=1 eqn}\\
x\in \{\gamma,a+\gamma,b+\gamma,a+b+\gamma\} &: \eqref{FBU_eqn} \ \Leftrightarrow\ \gamma^3+\gamma^2+(a^2+ab+b^2)\gamma=a^2b+ab^2+a^2+ab+b^2. \label{Inv1r_char2 x=r eqn}
\end{align}
If $x\not\in \{0,a,b,a+b\}\cup\{1,a+1,b+1,a+b+1\}\cup\{\gamma,a+\gamma,b+\gamma,a+b+\gamma\}$, then \eqref{FBU_eqn} is equivalent to $ab(a+b)=0$, which implies trivial cases. So, \eqref{FBU_eqn} has at most 12 solutions in $\{0,a,b,a+b\}\cup\{1,a+1,b+1,a+b+1\}\cup\{\gamma,a+\gamma,b+\gamma,a+b+\gamma\}$. Now we analyze that  \eqref{FBU_eqn} has at least 8 solutions. \\
(i) Substituting \eqref{Inv1r_char2 x=0 eqn} to \eqref{Inv1r_char2 x=1 eqn}, we have $\gamma+1=a^2b+ab^2$. By \eqref{Inv1r_char2 x=0 eqn}, we have $\frac b a \in \F_4 \setminus \F_2$ and hence we obtain $\gamma+1=a^3=b^3$.\\
(ii) Substituting \eqref{Inv1r_char2 x=0 eqn} to \eqref{Inv1r_char2 x=r eqn}, we have $\gamma^3+\gamma^2=a^2b+ab^2$. Similarly with (i) we obtain $\frac b a \in \F_4 \setminus \F_2$ and hence $\gamma^3+\gamma^2=a^3=b^3$.\\
(iii) Adding \eqref{Inv1r_char2 x=1 eqn} and \eqref{Inv1r_char2 x=r eqn}, we have $(\gamma+1)^3=0$, 	
which contradicts the assumption that $\gamma\ne 1$.  Hence \eqref{Inv1r_char2 x=1 eqn} and \eqref{Inv1r_char2 x=r eqn} cannot hold simultaneously.

By (iii), \eqref{FBU_eqn} cannot have $12$ solutions, and hence $\nabla_f(a,b)\in \{0,4,8\}$. By (i) and (ii), we have $\nabla_f(a,b)=8$ if and only if $\frac b a \in \F_4\setminus \F_2$ and $a^3=b^3=\gamma+1$ or $a^3=b^3= \gamma^3+\gamma^2$ with $1,\gamma \not \in \{a,b,a+b\}$. Moreover, we have $\nabla_f(a,b)=4$ if only one of \eqref{Inv1r_char2 x=0 eqn}, \eqref{Inv1r_char2 x=1 eqn} and \eqref{Inv1r_char2 x=r eqn} holds. Hence, $\nabla_f(a,b)=4$ if and only if  $\frac b a \in \F_4 \setminus \F_2$ and both \eqref{Inv1r_char2 x=1 eqn} and \eqref{Inv1r_char2 x=r eqn} do not hold, or $\frac b a \not\in \F_4 \setminus \F_2$ and one of \eqref{Inv1r_char2 x=1 eqn} and \eqref{Inv1r_char2 x=r eqn} holds. 

Therefore, under the assumption  that $1,\gamma \not \in \{a,b,a+b\}$, we have
\begin{equation*}
\nabla_f(a,b)=
\begin{cases}
8, &\text{ if }\frac b a \in\F_4\setminus \F_2, a^3\in \{\gamma+1,\gamma^3+\gamma^2\}, \\
4, &\text{ if }\frac b a \not\in \F_4, (a,b)\in S_{\gamma,7}\cup S_{\gamma,8},\\
&\text{ or }\frac b a \in \F_4\setminus\F_2, (a,b)\not\in S_{\gamma,7}\cup S_{\gamma,8}, \\
0, &\text{ otherwise.}
\end{cases}
\end{equation*}
We complete the proof.
\end{proof}

Next, we further investigate the second-order zero differential spectrum of $f=Inv \circ (1,\gamma)$. We denote
\begin{equation*}
\begin{array}{ll}
A_4=\{1\}\cup S_{\gamma,1},\  &S_4 =\{(a,b)\in A_4\times A_4 : a\ne b\}, \\
A_5=\{\gamma\}\cup S_{\gamma,2},\ &S_5 =\{(a,b)\in A_5\times A_5 : a\ne b\}.
\end{array}
\end{equation*}
By Theorem \ref{Inv1r_char_even_thm}, we have $\nabla_f(a,b)=8$ if $(a,b)\in S_4 \cup S_5$ when $\gamma\in \F_8$.
By Lemma \ref{quad lemma}, $x^2+x+\gamma=0$ has two solutions in $\F_8 \setminus \F_2$ if and only if $\tr(\gamma)=0$, and $x^2+\gamma x+\gamma=0$ has two solutions in $\F_8 \setminus \F_2$ if and only if $\tr\left( \frac1{\gamma}\right)=0$. When $\gamma\in\F_8\setminus \F_2$, we have $\tr\left(\frac{1}{\gamma}\right)=\tr(\gamma^3)$ and 
\begin{equation}\label{F8_trace_r}
\tr\left(\gamma+\frac{1}{\gamma}\right)=\tr(\gamma+\gamma^3) = (\gamma+\gamma^2+\gamma^{2^2})+(\gamma^3+\gamma^{3\cdot 2}+\gamma^{3\cdot 2^2})
=\gamma+\gamma^2+\gamma^3+\gamma^4+\gamma^5+\gamma^6= 1.
\end{equation}
So, we have $\#A_4=3$ and $\#A_5=0$, or $\#A_4=0$ and $\#A_5=3$, and hence $\#(S_4 \cup S_5)=6$.

When $2\mid n$, we have $\nabla_f(a,b)=8$ if $(a,b)\in S_6 \cup S_7$ where
\begin{align*}
A_6=\{a\in \Fbnmul : a^3=\gamma+1\},\  &S_6 =\{(a,b)\in A_6\times A_6 : a\ne b, \}, \\
A_7=\{a\in \Fbnmul : a^3=\gamma^3+\gamma^2\},\ &S_7 =\{(a,b)\in A_7\times A_7 : a\ne b, \}.
\end{align*}
We consider the cases that $(a,b)\in S_6 \cup S_7$ and $\{a,b,a+b\}\cap\{1,\gamma\}\ne\emptyset$. We assume that $a\in \{1,\gamma\}$, and the cases $b, a+b\in \{1,\gamma\}$ are similar. If $a=1$, then we have $\gamma^3+\gamma^2=1$, since $\gamma\ne0$. If $a=\gamma$, then $\gamma+1=\gamma^3$, , since $\gamma\ne0$. Note that $\gamma\in \F_8\setminus \F_2$ if and only if $\gamma^3+\gamma+1=0$ or $\gamma^3+\gamma^2+1=0$. Therefore, we have $\omega_8 = 12$ when $6\mid n$ and $\gamma\in\F_8\setminus \F_2$. 

If $\gamma\not \in \F_8$, then we have $\nabla_f(a,b)=8$ for all $(a,b)\in  S_6\cup S_7$. And we also have $\#(S_6\cup S_7) = 6\cdot \#(C_n \cap \{\gamma+1, \gamma^3+\gamma^2\})$, where 
\begin{equation}\label{set of cubic}
C_n=\{\alpha^3 : \alpha \in \Fbn\setminus \F_4\} 
\end{equation}
is the set of all cubic elements in $\Fbn\setminus\{0,1\}$.

Therefore, we have the following
\begin{equation}\label{Inv1r_DS8}
\omega_8 =
\begin{cases}
12, &\text{ if }6\mid n, \gamma\in\F_8\setminus \F_2\\
6, &\text{ if }3\mid n, 2\nmid n, \gamma\in\F_8\setminus \F_2\\
6\cdot \#(C_n \cap \{\gamma+1, \gamma^3+\gamma^2\}), &\text{ if } 2\mid n, \gamma\not\in\F_8,\\
0, &\text{ otherwise,}
\end{cases}
\end{equation}
where $C_n$ is given in \eqref{set of cubic}. In particular, we have the following corollary from \eqref{Inv1r_DS8}.

\begin{corollary} Let $p=2$ and $f(x) = Inv\circ(1,\gamma)$ where $\gamma \in \Fbn\setminus \F_2$. Then,
\begin{equation*}
\nabla_f=
\begin{cases}
8, &\text{ if }2\mid n, \{\gamma^3+\gamma^2,\gamma+1\}\cap C_n\ne \emptyset,\\
& \ \ \ \ 3\mid n, \gamma\in \F_8\setminus \F_2,\\
4, &\text{ otherwise,}
\end{cases}
\end{equation*}
where $C_n$ is given in \eqref{set of cubic}. In particular, we have $\nabla_f=4$ for all $\gamma\in \Fbn$, if $n$ is odd with $3\nmid n$. 
\end{corollary}

Next we consider to compute $\omega_4$. 

\bigskip
\noindent (Case 1) When $n$ is even, there are 6 pairs $(a,b)$ such that $a,b\in\F_4^*$, $a\ne b$ for $\gamma\in \F_4\setminus \F_2$.

\bigskip
\noindent (Case 2) We consider to compute the number of pairs $(a,b)$ such that 
\begin{align*}
1\in\{a,b,a+b\},&\  (\{a,b\}\setminus\{1\})\subset S_{\gamma,3}\cup S_{\gamma,4}, \text{ or }\\
\gamma\in\{a,b,a+b\},&\ (\{a,b\}\setminus\{\gamma\})\subset S_{\gamma,5}\cup S_{\gamma,6}.
\end{align*}
We denote
\begin{align*}
A_8&=\{1\}\cup S_{\gamma,3},\  A_9=\{1\}\cup S_{\gamma,4},\ A_{10}=\{\gamma\}\cup S_{\gamma,5},\ A_{11}=\{\gamma\}\cup S_{\gamma,6},\\
S_k  &=\{(a,b)\in A_k\times A_k : a\ne b\},\ k=8,9,10,11.
\end{align*}
By Theorem \ref{Inv1r_char_even_thm}, we have $\nabla_f(a,b)=4$ if $(a,b)\in S_8 \cup S_9 \cup S_{10}\cup S_{11}$. By Lemma \ref{quad lemma}, 
\begin{equation*}
\begin{array}{ll}
\#S_{\gamma,3}=2\ \Leftrightarrow \ \tr(\gamma)=0,\ & \ \#S_{\gamma,4} = 2\ \Leftrightarrow \ \tr(1/\gamma)=\tr(1),\\
\#S_{\gamma,5}=2\ \Leftrightarrow \ \tr(1/\gamma)=0,\ & \ \#S_{\gamma,6} = 2\ \Leftrightarrow \ \tr(\gamma)=\tr(1),
\end{array}
\end{equation*}
assuming $\gamma \in\Fbn \setminus \F_8$. If $n$ is even, we have
\begin{equation*}
\#(S_8\cup S_{11}) = 
\begin{cases}
12, & \text{ if }\tr(\gamma)=0,\\
0, & \text{ if }\tr(\gamma)=1,
\end{cases}
 \ \ \ \#(S_9\cup S_{10}) = 
\begin{cases}
12, & \text{ if }\tr(1/\gamma)=0,\\
0, & \text{ if }\tr(1/\gamma)=1.
\end{cases}
\end{equation*}
If $n$ is odd, then we have $\#(S_8\cup S_{11}) =  \#(S_9\cup S_{10})=6$. Note that $\#(S_8\cup S_9\cup S_{10}\cup S_{11})=0$ when $\gamma\in \F_8$.

\bigskip
\noindent (Case 3) We consider to compute the number of pairs $(a,b)$ such that
$$\frac b a \in \F_4\setminus\F_2, (a,b)\not\in S_{\gamma,7}\cup S_{\gamma,8}, \{a,b,a+b\}\cap\{1,\gamma\}=\emptyset.$$
 If $\frac b a \in \F_4 \setminus \F_2$ then \eqref{Inv1r_char2 x=1 eqn} is equivalent to $a^3=b^3=\gamma+1$, and \eqref{Inv1r_char2 x=r eqn} is equivalent to $a^3=b^3=\gamma^2(\gamma+1)$. So, there are $6\cdot \#(C_n \setminus \{\gamma^3,\gamma+1, \gamma^3+\gamma^2\})=6[(2^n-1)/3 -1] -6\cdot \#(C_n \cap \{\gamma^3,\gamma+1, \gamma^3+\gamma^2\})= 2^{n+1}-8-6\cdot \#(C_n \cap \{\gamma^3,\gamma+1, \gamma^3+\gamma^2\})$ pairs $(a,b)$ with $a^3=b^3\ne \gamma+1, \gamma^2(\gamma+1)$ when $n$ is even. Recall that we have $\nabla_f(a,b)=8$ when $a^3=b^3\ne \gamma+1, \gamma^2(\gamma+1)$, hence these pairs are also excluded when we count the number of pairs $(a,b)\in (S_{\gamma,7}\cup S_{\gamma,8})\setminus S$ with $\nabla_f(a,b)=4$. We also note that $\tr(\gamma)=\tr(1/\gamma)=0$ for $\gamma\in\F_8\setminus \F_2$, when $6\mid n$.

\bigskip
Therefore, we have
\begin{align*}
\omega_4=&\#((S_{\gamma,7}\cup S_{\gamma,8})\setminus S)-12\cdot \#(C_n \cap \{\gamma^3,\gamma+1, \gamma^3+\gamma^2\})\\
+&
\begin{cases}
2^{n+1}+22,&\text{ if } \tr(\gamma)=\tr(1/\gamma)=0,\\
2^{n+1}+10,&\text{ if }  \tr(\gamma+1/\gamma)=1,\\
2^{n+1}-2,&\text{ if }  \tr(\gamma)=\tr(1/\gamma)=1,
\end{cases}
\end{align*}
when $n$ is even, and 
\begin{equation*}
\omega_4=\#((S_{\gamma,7}\cup S_{\gamma,8})\setminus S)+
\begin{cases}
12,&\text{ if } \gamma\not \in \F_8, \\
0,&\text{ if } \gamma \in \F_8, \\
\end{cases}
\end{equation*}
when $n$ is odd, where 
$$S=\{(a,b) \in \Fbnmul \times \Fbnmul : \{1,\gamma\}\cap \{a,b,a+b\}\ne \emptyset\}.$$

In particular, when $n$ is odd, we have
\begin{equation}\label{Inv1r_DS_odd}
\begin{cases}
\omega_8=6,\ \ \omega_4=\#((S_{\gamma,7}\cup S_{\gamma,8})\setminus S), &\text{ if }3\mid n, \gamma \in \F_8\setminus \F_2,\\
\omega_8=0,\ \ \omega_4=\#((S_{\gamma,7}\cup S_{\gamma,8})\setminus S)+12, &\text{ otherwise.}
\end{cases}
\end{equation}
We desire to further investigate $\omega_4$, when $n$ is odd. We consider to compute the number of pairs $(a,b)$ such that 
$\frac b a \not\in \F_4, (a,b)\in S_{\gamma,7}\cup S_{\gamma,8}, \{a,b,a+b\}\cap\{1,\gamma\}=\emptyset.$
We rewrite \eqref{Inv1r_char2 x=1 eqn} as
\begin{equation}\label{Inv1r_char2 x=1 eqnb}
b^2(a+\gamma+1)+ab(a+\gamma+1)+(a+1)^2(\gamma+1)=0.
\end{equation}
By Lemma \ref{quad lemma}, there are two $b\in \Fbn$ satisfying \eqref{Inv1r_char2 x=1 eqnb} if and only if 
\begin{align*}
0&=\tr\left(\Frac{(a+1)^2(\gamma+1)}{a^2(a+\gamma+1)}\right) = \tr\left(\Frac{(a+1)^2}{a^2}\right) + \tr\left(\Frac{a(a+1)^2}{a^2(a+\gamma+1)}\right)\\
&=\tr\left(\Frac{a+1}{a}\right) + \tr\left(\Frac{a^2+1}{a(a+\gamma+1)}\right)=\tr\left(\Frac{(a+1)(a+\gamma+1)+a^2+1}{a(a+\gamma+1)}\right)\\
&=\tr\left(\Frac{\gamma(a+1)}{a(a+\gamma+1)}\right)=\tr\left(\Frac{\gamma}{a+\gamma+1}\right)+\tr\left(\Frac{\gamma}{a(a+\gamma+1)}\right).
\end{align*}
Hence, there are two $b\in \Fbn$ satisfying \eqref{Inv1r_char2 x=1 eqn} if and only if 
\begin{equation}\label{Sr7_trace}
\tr\left(\Frac{\gamma}{a+\gamma+1}\right)=\tr\left(\Frac{\gamma}{a(a+\gamma+1)}\right).
\end{equation}
Similarly, there are two $b\in \Fbn$ satisfying \eqref{Inv1r_char2 x=r eqn} if and only if 
\begin{equation}\label{Sr8_trace}
\tr\left(\Frac{1}{a+\gamma+1}\right)=\tr\left(\Frac{\gamma}{a(a+\gamma+1)}\right).
\end{equation}
If we set
$$S_{\gamma,9}=\left\{a\in \Fbnmul : \eqref{Sr7_trace}\text{ holds}, a\ne \gamma+1\right\}\ \text{ and } \ S_{\gamma,10}=\left\{a\in \Fbnmul : \eqref{Sr8_trace}\text{ holds}, a\ne \gamma+1\right\},$$
then for every $a\in \Fbn\setminus\{0,1,\gamma,\gamma+1\}$,
\begin{equation*}
\#\{b\in \Fbn : (a,b)\in S_{\gamma,7}\cup S_{\gamma,8}\}=
\begin{cases}
4, &\text{ if }a\in S_{\gamma,9}\cap S_{\gamma,10},\\
2, &\text{ if }a\in (S_{\gamma,9}\setminus S_{\gamma,10})\cup(S_{\gamma,10}\setminus S_{\gamma,9}),\\
0, &\text{ if }a\not\in S_{\gamma,9}\cup S_{\gamma,10}.
\end{cases}
\end{equation*}
Since the right hand sides of \eqref{Sr7_trace} and \eqref{Sr8_trace} are identical, $a\in (S_{\gamma,9}\setminus S_{\gamma,10})\cup (S_{\gamma,10}\setminus S_{\gamma,9}) $ is equivalent to $\tr\left(\frac{\gamma+1}{a+\gamma+1}\right)=1$, and then we have $2$ pairs $(a,b)$ in $S_{\gamma,7}\cup S_{\gamma,8}$ in this case. Since the field trace map is balanced, the number of $a\in \Fbnmul \setminus \{\gamma+1\}$ satisfying $\tr\left(\frac{\gamma+1}{a+\gamma+1}\right)=1$ is $2^{n-1}-1$, when $n$ is odd.

\begin{remark}\label{Sr78_conjecture}
If $\tr\left(\frac{\gamma+1}{a+\gamma+1}\right)=0$ or equivalently $\tr\left(\frac{a}{a+\gamma+1}\right)=1$, then $a\in S_{\gamma,9}\cap S_{\gamma,10}$ or $a\not \in S_{\gamma,9}\cup S_{\gamma,10}$ holds. By SageMath experiment, $\#(S_{\gamma,9}\cap S_{\gamma,10})=2^{n-2}-1$ for any $\gamma\in \Fbn\setminus\{0,1\}$, when $1<n\le 17$ is odd. We conjecture that this holds for every odd $n>1$. 
\end{remark}
From now on, we assume that the conjecture in Remark \ref{Sr78_conjecture} is true. Then we have 
\begin{align}
\#((S_{\gamma,7}\cup S_{\gamma,8})\setminus S)&= 4\cdot \#(S_{\gamma,9}\cap S_{\gamma,10}) +2\cdot  \#\left((S_{\gamma,9}\setminus S_{\gamma,10})\cup(S_{\gamma,10}\setminus S_{\gamma,9})\right) -  \#(S\cap (S_{\gamma,7}\cup S_{\gamma,8})) \notag\\
&=4(2^{n-2}-1)+2(2^{n-1}-1) -  \#(S\cap (S_{\gamma,7}\cup S_{\gamma,8}))\notag \\
&= 2^{n+1}-6-  \#(S\cap (S_{\gamma,7}\cup S_{\gamma,8})).\label{Sr78_num}
\end{align}
It remains to compute $\#(S\cap (S_{\gamma,7}\cup S_{\gamma,8}))$, and hence we consider to compute the number of pairs $(a,b)\in \Fbnmul \times \Fbnmul$ with $\{1,\gamma\} \cap \{a,b,a+b\}\ne \emptyset$ such that \eqref{Sr7_trace} or \eqref{Sr8_trace}. If $a=1$, then \eqref{Sr7_trace} holds. But, $a=1$ in \eqref{Inv1r_char2 x=1 eqn} implies that $b=0$ or $b=1$, a contradiction. And, $a=\gamma$ in \eqref{Sr7_trace} is equivalent to $\tr(\gamma)=\tr(1)=1$. Hence we have 4 pairs of $(1,b)\in (S_{\gamma,7}\cup S_{\gamma,8}))$ when $\tr(\gamma)=1$, and 2 pairs otherwise. If $a=\gamma$, then \eqref{Sr8_trace} holds. But, $a=\gamma$ in \eqref{Inv1r_char2 x=r eqn} implies that $b=0$ or $b=\gamma$, a contradiction. And, $a=1$ in \eqref{Sr8_trace} is equivalent to $\tr\left(\frac 1\gamma\right)=\tr(1)=1$. Hence we have 4 pairs of $(\gamma,b)\in (S_{\gamma,7}\cup S_{\gamma,8}))$ when $\tr\left(\frac 1\gamma\right)=1$, and 2 pairs otherwise. 

Applying $b=1$ in \eqref{Inv1r_char2 x=1 eqn} implies that $a=0$ or $a=1$, which are already considered in above. Applying $b=\gamma$ in \eqref{Inv1r_char2 x=1 eqn} implies that $a^2+\gamma a +(\gamma+1)^3=0$. By Lemma \ref{quad lemma}, $a^2+\gamma a +(\gamma+1)^3=0$ has two solutions if and only if $\tr\left(\frac{(\gamma+1)^3}{\gamma^2}\right)=\tr(\gamma+1)=0$ or equivalently $\tr(\gamma)=1$. Since \eqref{Inv1r_char2 x=1 eqn} and \eqref{Inv1r_char2 x=r eqn} have no common solution, \eqref{Sr7_trace} holds and \eqref{Sr8_trace} does not hold for each solution $a$ of $a^2+\gamma a +(\gamma+1)^3=0$. Since if $b=\gamma$ is a solution of \eqref{Inv1r_char2 x=1 eqn} then another solution of \eqref{Inv1r_char2 x=1 eqn} is $b=a+\gamma$, so we have 2 pairs of $(a,\gamma)\in S_{\gamma,7}\cup S_{\gamma,8}$ and 2 pairs of $(a,a+\gamma)\in S_{\gamma,7}\cup S_{\gamma,8}$, when $\tr(\gamma)=1$. Similarly, $b=\gamma$ in \eqref{Inv1r_char2 x=r eqn} implies already considered case in above, and $b=1$ in \eqref{Inv1r_char2 x=r eqn} implies $\gamma a^2+\gamma a+(\gamma+1)^3=0$, which has two solutions if and only if $\tr\left(\frac{(\gamma+1)^3}{\gamma}\right)=\tr\left(1+\frac1\gamma\right)=0$ or equivalently $\tr\left(\frac1\gamma\right)=1$. \eqref{Sr8_trace} holds and \eqref{Sr7_trace} does not hold for each solution $a$ of $\gamma a^2+\gamma a+(\gamma+1)^3=0$. Hence, we have 2 pairs of $(a,1)\in S_{\gamma,7}\cup S_{\gamma,8}$ and 2 pairs of $(a,a+1)\in S_{\gamma,7}\cup S_{\gamma,8}$, when $\tr\left(\frac1\gamma\right)=1$.
To summarize, we have
\begin{equation*}
\#(S\cap (S_{\gamma,7}\cup S_{\gamma,8})) = 
\begin{cases}
4\cdot 0 + 2\cdot 2 = 4,& \text{ if } \tr(\gamma)=\tr\left(\frac 1 \gamma\right)=0,\\
4\cdot 4 + 2\cdot 0 = 16,& \text{ if } \tr(\gamma)=\tr\left(\frac 1 \gamma\right)=1,\\
4\cdot 2 + 2\cdot 1 = 10,& \text{ if } \tr\left(\gamma+\frac 1 \gamma\right)=1.
\end{cases}
\end{equation*}

So, if $3\mid n$ and $\gamma\in \F_8 \setminus \F_2$, then we have $\nabla_f=8$ and
$$\omega_8=6,\ \  \omega_4=2^{n+1}-16,\ \  \omega_0=(2^n-2)(2^n-1)-\omega_4-\omega_8 = 2^{2n}-5\cdot 2^n+12.$$
applying \eqref{F8_trace_r} on \eqref{Sr78_num}. Otherwise, we have $\nabla_f=4$, and applying \eqref{Inv1r_DS_odd} and \eqref{Sr78_num} we have
\begin{equation*}
\omega_4 = 2^{n+1}+6-  \#(S\cap (S_{\gamma,7}\cup S_{\gamma,8}))=
\begin{cases}
2^{n+1}+2,& \text{ if } \tr(\gamma)=\tr\left(\frac 1 \gamma\right)=0,\\
2^{n+1}-10,& \text{ if } \tr(\gamma)=\tr\left(\frac 1 \gamma\right)=1,\\
2^{n+1}-4,& \text{ if } \tr\left(\gamma+\frac 1 \gamma\right)=1,
\end{cases}
\end{equation*}
\begin{equation*}
\omega_0  = (2^n-1)(2^n-2) - \omega_4 = 
\begin{cases}
2^{2n}-5\cdot 2^n,& \text{ if } \tr(\gamma)=\tr\left(\frac 1 \gamma\right)=0,\\
2^{2n}-5\cdot 2^n+12,& \text{ if } \tr(\gamma)=\tr\left(\frac 1 \gamma\right)=1,\\
2^{2n}-5\cdot 2^n+6,& \text{ if } \tr\left(\gamma+\frac 1 \gamma\right)=1.
\end{cases}
\end{equation*}
We confirm that the above holds for all $\gamma\in \Fbn\setminus \F_2$ via SageMath experiments, when $1<n\le 11$ is odd.

\subsection{Some experimental results on the second-order zero differential uniformity in odd characteristic}

In this subsection, we consider the second-order zero differential uniformity of $f=Inv\circ (1,\gamma)$ when $p$ is odd, by using similar idea with Section \ref{subsec_inv01_FBU}. However, unlike in Section \ref{subsec_inv01_FBU} where $\gamma=0$ and $p=2$, $\gamma$ is not fixed and $p$ is an arbitrary odd prime in this section, so we have to deal with a much wider range of cases, which can be a very tedious process. However, since the inverse function has low second-order zero differential uniformity and $Inv \circ (1,\gamma)$ differs from the inverse function by only two points, we can conjecture that  $Inv \circ (1,\gamma)$ also has the low second-order zero differential uniformity. We introduce several experimental results using SageMath. When $p^n <1000$ with $p>3$ we have the followings 
\begin{itemize}
\item $\nabla_f=7$ when $p^n=11^2$ and $\gamma=-1$,
\item $\nabla_f=6$ when $p=5$, $\gamma^4=-1$ or $p=97$
\item otherwise, $2\le \nabla_{f}\le 5$.
\end{itemize}

On the other hand, when $p=3$, we have $\nabla_f \in \{3,6,9\}$. These results can be partially explained by investigating $\nabla_f(a,a)$ for $a\in \F_{3^n}^*$. Let $b=a$. Then \eqref{FBU_eqn} is equivalent to
\begin{equation}\label{p3_ab0_eqn}
f(x)+f(x+a)+f(x-a)=0
\end{equation}
When $a=1$,
\begin{itemize}
\item If $x\in \F_3$, then \eqref{p3_ab0_eqn} is equivalent to $\gamma=1$, a contradiction.
\item If $x\in \{\gamma, \gamma\pm1\}$, then \eqref{p3_ab0_eqn} is equivalent to $\gamma^2-\gamma=1$, such $\gamma$ exists only if $n$ is even.
\item If $x\not \in \{0,\pm1, \gamma, \gamma\pm1\}$, then \eqref{p3_ab0_eqn} is equivalent to $-a^2=0$, a contradiction. 
\end{itemize}
Therefore, $\nabla_f(1,1)=3$ when $\gamma^2-\gamma=1$, and  $\nabla_f(1, 1)=0$ otherwise.

When $a=\gamma$, 
\begin{itemize}
\item If $x\in \{1,1\pm \gamma\}$, then \eqref{p3_ab0_eqn} is equivalent to $\gamma^2+\gamma=1$, such $\gamma$ exists only if $n$ is even.
\item If $x\in \{0,\pm\gamma\}$, then \eqref{p3_ab0_eqn} is equivalent to $\gamma=1$, a contradiction.
\item If $x\not \in \{0,1, \pm\gamma, 1\pm\gamma\}$, then \eqref{p3_ab0_eqn} is equivalent to $-\gamma^2=0$, a contradiction. 
\end{itemize}
Therefore, $\nabla_f(\gamma, \gamma)=3$ when $\gamma^2+\gamma=1$, and  $\nabla_f(\gamma, \gamma)=0$ otherwise.

If $a\ne 1,\gamma$, then we have
\begin{itemize}
\item $x=0,\pm a$ are solutions of \eqref{p3_ab0_eqn}.
\item $x=1, \pm a+1$ are solutions of \eqref{p3_ab0_eqn} if $\gamma=1-a^2$.
\item $x=\gamma, \pm a+\gamma$ are solutions of \eqref{p3_ab0_eqn} if $\gamma(\gamma-1)=a^2$.
\item $x\not\in\{0,1,\gamma, \pm a, 1\pm a, \gamma \pm a\}$, then \eqref{p3_ab0_eqn} is equivalent to $-1=0$, a contradiction. 
\end{itemize}
We add $\gamma=1-a^2$ and $\gamma(\gamma-1)=a^2$, then we obtain $\gamma^2=1$. Since $\gamma\ne 1$, we have $\gamma=-1$, and $a^2=b^2=-1$. For $a\in \F_{3^n}$ with $a^2=-1$ to exist, $n$ need to be even. Therefore, we have $\nabla_f(a,  a)=9$, if $2\mid n, \gamma=-1, a^2=-1$.

Furthermore, we have $\nabla_f(a,a)= 6$, if $a^2=1-\gamma$ or $a^2=\gamma(\gamma-1)$, and $1-\gamma, \gamma(\gamma-1)\not\in\{1,\gamma^2\}$. By the way, $1-\gamma, \gamma(\gamma-1)\in\{1,\gamma^2\}$ implies that $\gamma^2\pm\gamma=1$, since $\gamma \ne 0$. 
If $\gamma^2+\gamma=1 \Leftrightarrow 1-\gamma=\gamma^2=a^2$ or $\gamma^2-\gamma=1=a^2$, then we already have $\nabla_f(a,a)=3$ in these cases.
Therefore, we have $\nabla_f (a,a)= 6$ if $a\in \{1-\gamma, \gamma(\gamma-1)\}\cap Q_n\ne \emptyset$ and $\gamma^2\pm \gamma\ne 1$, where 
$$Q_n=\{a^2\in \F_{3^n} : a\in \F_{3^n}\}\setminus \F_3$$
is the set of all squares in $\F_{3^n}\setminus \F_3$. To summarize, we have
\begin{equation*}
\nabla_f (a,a)= 
\begin{cases}
9, &\text { if }2\mid n, \gamma=-1, a^2=-1\\
6, &\text{ if }a^2\in \{1-\gamma, \gamma(\gamma-1)\}\cap Q_n \text{ and }\gamma^2\pm \gamma\ne 1,\\
3, &\text { otherwise.}
\end{cases}
\end{equation*}
Furthermore, note that we have the following :
\begin{equation*}
\nabla_f = 
\begin{cases}
9, &\text { if }2\mid n, \gamma=-1,\\
6, &\text{ if }\{1-\gamma, \gamma(\gamma-1)\}\cap Q_n\ne \emptyset \text{ and }\gamma^2\pm \gamma\ne 1,\\
3, &\text { otherwise,}
\end{cases}
\end{equation*}
from SageMath experiments for $2\le n \le 8$.  We conjecture that the above holds for every $n\ge 2$.

\section{Conclusion}\label{sec_con}

In this paper, we studied the second-order zero differential uniformity of the swapped inverse function $Inv \circ (1,\gamma)$ where $\gamma \in \Fpn\setminus \{1\}$. Specifically, 
\begin{itemize}
\item We studied the FBCT of $Inv \circ (0,1)$ and completely computed the second-order zero differential spectrum of $Inv \circ (0,1)$ defined on $\Fbn$.
\item We showed that the second-order zero differential uniformity of $Inv \circ (0,1)$ defined on $\Fpn$ is at most $4$, when $p$ is an odd prime.
\item We studied the FBCT of $Inv\circ (1,\gamma)$ defined on $\Fbn$, where $\gamma\in\Fbn\setminus \{0,1\}$.  Moreover, we further studied the second-order zero differential spectrum of $Inv \circ (1,\gamma)$, especially when $n$ is odd.
\item Due to the extensive range of cases, rather than conducting a detailed investigation into the second-order zero differential uniformity of $Inv \circ (1,\gamma)$ when $p$ is odd, we presented several experimental results for some cases of $\gamma$ and $p$. Furthermore, from experimental results, we suggested a conjecture on the second-order zero differential uniformity of $Inv \circ (1,\gamma)$ when $p=3$.
\end{itemize}
In summary, the swapped inverse functions, known to have good cryptographic properties, also have low second-order zero differential uniformity. In particular, our paper is the first to characterize classes of non-power functions with the second-order zero differential uniformity equal to $4$ over binary fields.

Finding the second-order zero differential uniformity of functions with good cryptographic properties has not been widely studied so far. Also, when restricted to non-power functions, the results are very limited. Therefore, we believe that this topic is worth further study, and we plan to explore the second-order zero differential uniformity of other non-power functions.

\bigskip

\noindent \textbf{Acknowledgements} :
This work was supported by the National Research Foundation of Korea (NRF) grant
funded by the Korea government (MSIT) (No. 2021R1C1C2003888). Soonhak Kwon was supported by the National Research Foundation of Korea (NRF) grant funded by the Korea government (MSIT) (No. 2016R1A5A1008055, No. 2019R1F1A1058920 and 2021R1F1A1050721).

\end{document}